\newif\ifpublic
\publictrue 

\documentclass[letterpaper,11pt,english]{article}

\usepackage{times}

\ifpublic
\usepackage[disable]{todonotes}
\else
\usepackage[color=green!20]{todonotes}
\fi

\usepackage{amssymb,amsmath,amsfonts,color,paralist}
\usepackage{amsthm}
\allowdisplaybreaks
\numberwithin{equation}{section}
\usepackage[top=1.1in,left=1.2in,right = 1.1in,bottom=1.1in]{geometry}
\usepackage{hyperref}

\newcommand{\footremember}[2]{%
    \footnote{#2}
    \newcounter{#1}
    \setcounter{#1}{\value{footnote}}%
}

\usepackage{algorithm}
\usepackage{algorithmic}
\usepackage{xspace}

\newtheorem{theorem}{Theorem}
\numberwithin{theorem}{section}
\newtheorem{lemma}[theorem]{Lemma}
\newtheorem{corollary}[theorem]{Corollary}
\newtheorem{proposition}[theorem]{Proposition}

\theoremstyle{definition}
\newtheorem{definition}[theorem]{Definition}

\theoremstyle{theorem}
\newtheorem*{remark}{Remark}
\newtheorem{observation}[theorem]{Observation}
\newtheorem{claim}[theorem]{Claim}

\newcommand{\R}{\mathbb{R}}
\newcommand{\eps}{\epsilon}
\newcommand{\suchthat}{\;:\;}

\newcommand{\indic}{\mathbb{I}}
\newcommand{\exx}{\mathcal{X}}
\renewcommand{\tilde}[1]{\widetilde{#1}}

\newcommand{\inparen}[1]{\left( #1 \right)}
\newcommand{\inbrace}[1]{\left\{ #1 \right\}}
\newcommand{\eqdef}{\triangleq}
\newcommand*{\coloneq}{\mathrel{\vcenter{\baselineskip0.5ex \lineskiplimit0pt
                     \hbox{\scriptsize.}\hbox{\scriptsize.}}}%
                     =}

\newcommand{\abs}[1]{\left|#1\right|}

\newcommand{\norm}[1]{\left\|#1\right\|}
\newcommand{\fnorm}[1]{\left\|#1\right\|_{F}}

\newcommand{\sr}[1]{\operatorname{sr}\left(#1\right)}
\newcommand{\rk}[1]{\operatorname{rank}\left (#1\right)}
\newcommand{\ssr}{\operatorname{ssr}}
\newcommand{\diam}{\operatorname{diam}}
\newcommand{\Frechet}{Fr\'{e}chet\xspace}

\title{Approximating Sparsest Cut in Low Rank Graphs via Embeddings from Approximately Low Dimensional Spaces.}
\author{Yuval Rabani\footremember{huji}{The Hebrew University of Jerusalem, Israel. e-mail: \texttt{yrabani@cs.huji.ac.il}.}\hspace{2cm}
Rakesh Venkat\footremember{huji2}{The Hebrew University of Jerusalem, Israel. e-mail: \texttt{rakesh@cs.huji.ac.il}. Supported by an I-Core Algorithms Fellowship.}}

\date{}
\begin{document}


\maketitle

\begin{abstract}
We consider the problem of embedding a finite set of points  $\inbrace{x_1, \ldots, x_n} \in \R^d$ that satisfy $\ell_2^2$ triangle inequalities into $\ell_1$, when the points are \emph{approximately} low-dimensional. Goemans (unpublished, appears in~\cite{MagenM2008}) showed that  such points residing in \emph{exactly} $d$ dimensions can be embedded into $\ell_1$ with distortion at most $\sqrt{d}$. We prove the following robust analogue of this statement: if  there exists a $r$-dimensional subspace $\Pi$ such that the projections onto this subspace satisfy $\sum_{i,j \in [n]}\norm{\Pi x_i - \Pi x_j}_2^2 \geq \Omega(1) \sum_{i,j \in [n]}\norm{x_i - x_j}_2^2$, then there is an embedding of the points into $\ell_1$ with $O(\sqrt{r})$ average distortion. A consequence of this result is that the integrality gap of the well-known Goemans-Linial SDP relaxation for the Uniform Sparsest Cut problem is $O(\sqrt{r})$ on graphs $G$ whose $r$-th smallest normalized eigenvalue of the Laplacian  satisfies $\lambda_r(G)/n \geq \Omega(1)\Phi_{SDP} (G)$. Our result improves upon the previously known bound of $O(r)$ on the average distortion, and the integrality gap of the Goemans-Linial SDP under the same preconditions, proven in \cite{DeshpandeV2014, DHV16}. 
\end{abstract}

\thispagestyle{empty}
\newpage
\setcounter{page}{1}

\section{Introduction}\label{sec:introduction}
A finite metric space consists of a pair $(\exx, d)$, where $\exx$ is a finite set of points, and $d:\exx \times \exx \rightarrow \R_{\geq0}$ is a distance function on pairs of points in $\exx$. Many combinatorial optimization problems can be naturally formulated as a maximization or minimization problem over metric spaces $(\exx, d)$ of some target class. However, since it might be computationally difficult to optimize over this class, one considers a \emph{relaxation} that finds a solution  $(\mathcal{Y}, d')$ amongst a class of computationally `easy' metrics, and then looks to produce an \emph{embedding} $\mathcal{Y} \hookrightarrow \exx $ into the target space, while minimizing some measure of \emph{distortion} between the distance functions $d$ and $d'$ incurred by the embedding. There has been much work that investigates various measures and costs   of distortion incurred by embeddings between metric spaces, and applications thereof (see the surveys \cite{IndykM2004,Matousek2002,Linial2002} and references therein).

\smallskip
In this work, we look at embeddings from $\ell_2^2$ metrics to $\ell_1$ metrics, motivated by applications to the Sparsest Cut problem. A $\ell_1$ metric (or a $\ell_1$ space) consists of  a finite set of points represented in $\R^d$ with the distance given by the $\ell_1$ distance between them. It is a natural target space that can be viewed as an non-negative combination of `cut-metrics' on the underlying point set, and hence arises frequently in graph-cut based problems. A $\ell_2^2$ space, on the other hand, is easy to optimize over, and consists of a finite set of points, say  $X = \inbrace{x_1, \ldots, x_n} \subset \R^d$, that satisfy triangle inequalities on the \emph{squares} of distances:
\begin{equation} \label{eq:l2-squared}
\norm{x_{i}-x_{j}}_2^2 +\norm{x_{j}-x_{k}}_2^2\geq \norm{x_{i}-x_{k}}_2^{2} \qquad \forall~ i,j,k \in [n].
\end{equation}

\smallskip
The Sparsest Cut problem is a fundamental NP-hard graph optimization problem that serves as a striking example of the utility of the metric embedding approach. In the (Uniform) Sparsest Cut problem, we are given a graph $G=(V,c)$, with a symmetric weight function $c_{ij}$ on pairs $\{i,j\}$. The goal is to find a cut $(S, \overline{S})$ of minimum \emph{sparsity} $\Phi(S)$, defined as follows (here, $\indic_{S}(i)$ is $1$, if $i \in S$, and $0$ otherwise). 
$$ \Phi(S) \coloneq \frac{\sum_{i<j} c_{ij} \abs{\indic_{S}(i) -
\indic_{S}(j)}}{\sum_{i<j} \abs{\indic_{S}(i) - \indic_{S}(j)}} $$

The best known approximation for the Sparsest Cut problem is due to Arora, Rao and Vazirani~\cite{AroraRV2009} (henceforth called the ARV algorithm), who considered the following semidefinite programming relaxation (SDP) introduced by Goemans and Linial (see \cite{Goemans1997} and \cite{Linial2002}).

\begin{align*}
\textbf{SDP-1:}\quad \Phi_{SDP}(G) & \coloneq \min_{\inbrace{x_i}_{i\in [n]}} ~ \frac{1}{n^2}\sum_{ij} c_{ij} \norm{x_{i}-x_{j}}_2^2 \\
&\text{s.t} \quad \begin{cases}\norm{x_{i}-x_{j}}_2^2 +\norm{x_{j}-x_{k}}_2^2\geq \norm{x_{i}-x_{k}}_2^{2} & \forall i,j,k\in[n].\\
{\sum_{kl} \norm{x_{k}-x_{l}}_2^2} = n^2. &
\end{cases}
\end{align*}

Clearly, $\Phi_{SDP}(G) \leq \Phi(G)$. Notice that any feasible solution to the above SDP constitutes a $ \ell_2 ^2 $ space. The ARV algorithm works by producing an embedding of the solutions of the above SDP into a $\ell_1$ space, with \emph{average distortion} (see Section~\ref{sec:notation} for a definition) $O(\sqrt{\log n})$. It was shown  in \cite{LinialLR1995, AumannR1998} that producing an embedding of the SDP solutions into a $\ell_1$ space with average distortion $D$ suffices to get a $O(D)$ approximation to the Uniform Sparsest Cut problem.

\smallskip
Though the solutions to SDP-1 can lie in up to $n$ dimensions, for certain graph classes, they are more structured. In particular, if the $r$-th smallest eigenvalue of the graph Laplacian satisfies $\lambda_r(G)/n \gg \Phi_{SDP}(G)$, then it turns out that the solutions are \emph{approximately} $r$-dimensional (see Definition~\ref{def:subspace-rank} and Section~\ref{sec:sparsestcut}). Graphs whose $r$-th smallest eigenvalue is bounded away from $0$ for a typically small $r$ are called \emph{low threshold-rank} graphs; note that spectral expanders are a special case of these for $r=2$. The work of Guruswami and Sinop~\cite{GuruswamiS2013} exploited higher levels of the Lasserre SDP hierarchy~\cite{Lasserre2001}, along with the above structure, to give constant-factor guarantees for Sparsest Cut on these graphs. However, this involved partially solving a SDP of size $n^{O(r)}$\footnote{In a separate work, Guruswami and Sinop~\cite{GuruswamiS2012b} give an algorithm that solves the SDP partially, running in $2^{O(r)} \mathrm{poly}(n)$ time, and suffices for their algorithm.}, and did not say anything about the behaviour of the Goemans-Linial SDP on these graphs.

\smallskip
Goemans showed that if the points satisfying $\ell_{2}^{2}$ triangle inequalities lie in $d$ dimensions, then they can be embedded into $\ell_2$ (and hence into $\ell_1$, since there is an isometry from $\ell_2$ to $\ell_1$ \cite{Matousek2002}) with $\sqrt{d}$ distortion  (unpublished, appears in~\cite{MagenM2008}, see also \cite[Section 4]{DHV16} for an alternative proof).
\begin{theorem}[{Goemans~\cite[Appendix~B]{MagenM2008}}] \label{thm:goemans-intro}
Let $x_{1}, x_{2}, \ldots, x_{n} \in \R^{d}$ be $n$ points satisfying $\ell_{2}^{2}$
triangle inequalities. Then there exists an embedding of these points into $\ell_2$,
$x_{i} \mapsto f(x_{i})$, with distortion $\sqrt{d}$, that is,
$$
\frac{1}{\sqrt{d}}~ \norm{x_{i} - x_{j}}_{2}^{2} \leq \norm{f(x_{i}) -
f(x_{j})}_{2} \leq \norm{x_{i} - x_{j}}_{2}^{2}, \quad \forall~ i, j\in
V.
$$
\end{theorem}

\smallskip
The immediate question that this raises is the following: can one reduce the dimension of $\ell_2^2$ metrics, while preserving pairwise distances, \emph{and} the $\ell_2^2$ triangle inequalities? The Johnson-Lindenstrauss lemma~\cite{JohnsonL1982}  reduces the dimension to $O(\log n)$, while preserving pairwise distances approximately. However, this procedure does not preserve the $\ell_2^2$ triangle inequalities, if the original points satisfied them. In fact, Magen and Moharammi~\cite{MagenM2008} prove a strong lower bound against dimension reduction for $\ell_2^2$ metrics. 

\smallskip
It is interesting to note that the Johnson-Lindenstrauss lemma, while not preserving the $\ell_2^2$ triangle inequalities exactly, does preserve them \emph{approximately}, that is, every sequence of $k \leq n$ points $x_{i_1}, \ldots, x_{i_k}$ satisfies $\sum_{j=1}^{k-1} \norm{x_{i_{j}} - x_{i_{j+1}}}_2^2 \geq \beta \cdot \norm{x_{i_1} - x_{i_k}}_2^2$, for some $\beta = \Omega(1)$. An observation by Luca Trevisan (personal communication) shows that, in fact, Goemans' theorem is also true for points satisfying approximate triangle inequalities, but the proof uses the ARV machinery. However, even this does not yield anything better that $O(\sqrt{\log n})$, for approximately $r$-dimensional points, when $r$ is small.

\smallskip
The above discussion motivates one to ask if there is a more `robust' analogue of Goemans' theorem that can be applied to low threshold-rank graphs. Deshpande, Harsha and Venkat~\cite{DHV16} considered this question, and showed that one can prove a similar theorem for the case where the points are in approximately $r$ dimensions, albeit giving a bound of $O(r)$ on the \emph{average} distortion (which suffices for Sparsest Cut). One would expect an exact analogue to have a bound of $O(\sqrt{r})$, and it was left open if one could find such an embedding. 

\smallskip 
We show that there is, indeed, an embedding into $\ell_1$ (in fact, into $\ell_2$, since all our embeddings are one-dimensional) with $O(\sqrt{r})$ average distortion when the points are approximately $r$-dimensional.


\subsection{Our Results}\label{sec:our-results}
In order to state our main result, we use the following definition to quantify the notion of approximate rank of a set of points: 

\begin{definition}($\eta$-Subspace rank) \label{def:subspace-rank}
For any $\eta \in (0,1]$, a set of points $X = \inbrace{x_1, \ldots, x_n} \subseteq \R^d $ will be said to have $\eta$-subspace rank $r$, denoted by $\ssr_\eta(X)=r$, if there exists a subspace given by a projector $\Pi \in \R^{d\times d}$ with $\rk{\Pi} =r $ that satisfies:
\begin{equation}\label{eq:pi-rank}
\sum_{i,j \in [n]} \norm{\Pi x_i - \Pi x_j}_2^2 ~\geq~ \eta \sum_{i,j \in [n]}\norm{x_i - x_j}_2^2 .
\end{equation}
In this work, we will always consider $\eta = \Omega(1)$.
\end{definition}

\begin{remark}\label{rem:top-r}
Since the subspace $\Pi_r$ defined by the top-$r$ left singular vectors of the matrix $M$ with columns $\{x_i - x_j\}_{ij}$ satisfies $\fnorm{\Pi_r M}^2 \geq \fnorm{\tilde{\Pi} M}^2$ for every $\tilde{\Pi}$ with $\rk{\tilde{\Pi}} \leq r $, we can always assume that $\Pi = \Pi_r(M)$ when we need to explicitly use the projections. Also, note that the subspace rank is independent of any scaling or shifting of the points, and is always at most the rank of the point set.
\end{remark}

\smallskip
Deshpande et al.~\cite{DHV16} use a slightly different notion of approximate dimension, called the \emph{stable-rank} of the point set, defined as $\sr{M} = \fnorm{M}^2 / \sigma_1(M)^2$, where $\sigma_1$ is the maximum singular value of the matrix $M$. Clearly, $\sr{M} \leq \ssr_\eta(X)/\eta$, and so points with low subspace rank also have low stable rank. While the stable rank is a well-known proxy for rank (see ~\cite{BourgainT1987,Tropp2009}), for applications to the Sparsest Cut problem, the notion of subspace rank suffices and is natural (see Section~\ref{sec:sparsestcut}). For applications to the Sparsest Cut problem, the notion of subspace rank suffices and is natural (see Section~\ref{sec:sparsestcut}). It would be interesting to see if other notions of approximate rank yield further applications or improvements, in Sparsest Cut, or elsewhere.

\smallskip

Our main result is the following:

\begin{theorem}\label{thm:main-thm}
Given a set of points $X = \inbrace{x_1, \ldots, x_n} \in \R^d$ with $\ssr_\eta(X)=r$ that satisfy the $\ell_2^2$ triangle inequalities, there is an embedding $ X \hookrightarrow \ell_1$ with average distortion at most $O_\eta (\sqrt{r})$. That is, there is a constant $c(\eta)$ and a mapping $h:X \rightarrow \R^{d'}$ that satisfies:
\begin{align}
\norm{h(x_i) - h(x_j)}_1 &~\leq~ \norm{x_i - x_j}_2^2  \qquad \forall i,j \in [n]\\ 
\sum_{i,j \in [n]} \norm{h(x_i) - h(x_j)}_1 &~\geq~ \frac{c(\eta)}{\sqrt{r}}\cdot\sum_{ij} \norm{x_i - x_j}_2^2
\end{align}
\end{theorem}

\smallskip

This matches Goemans' theorem in terms of the dependence on $r$, albeit for average-case distortion.  Since the subspace rank is an \emph{average} global condition on the point set, we cannot hope to prove a worst-case distortion guarantee like Goemans' theorem that depends only on the subspace rank (see Appendix~\ref{sec:no-worst-case}). 

\smallskip

The above theorem holds even if the points satisfy the $\ell_2^2$ triangle inequalities only \emph{approximately}, since the steps in the analysis of the algorithm only need the points to satisfy an approximate version of the triangle inequalities\footnote{The points are said to satisfy \emph{approximate} $\ell_2^2$ triangle inequalities, if every sequence of $k \leq n$ points $x_{i_1}, \ldots, x_{i_k}$ satisfies $\sum_{j=1}^{k-1} \norm{x_{i_{j}} - x_{i_{j+1}}}_2^2 \geq \beta \cdot \norm{x_{i_1} - x_{i_k}}_2^2$, for some $\beta = \Omega(1)$}. Improving on the $\sqrt{r}$ bound above with any technique that works with approximate triangle inequalities would imply an improvement over the ARV algorithm's guarantee, since dimension reduction using the Johnson-Lindenstrauss~\cite{JohnsonL1982} transform preserves pairwise distances (and hence the $\ell_2^2$ inequalities) approximately, while reducing the dimension to $O(\log n)$. Note that this, thus, recovers the unconditional guarantee of $O(\sqrt{\log n})$ of the ARV algorithm, but gives better results for points in lower approximate dimension. This is unsurprising, since our techniques do build on the ARV analysis.

\medskip

Our main result immediately implies a $O(\sqrt{r})$ approximation algorithm for the Uniform Sparsest Cut problem on low threshold-rank graphs, using just the Goemans-Linial SDP.

\begin{corollary}\label{cor:main-cor}
Let $\epsilon \in (0,1]$. Given a regular graph $G$ with $r$-th smallest eigenvalue of the normalized Laplacian satisfying $\lambda_r(G) \geq \Phi_{SDP}(G)/(1-\epsilon)$, we can find a $O_{\epsilon}(\sqrt{r})$ approximation to the sparsest cut in the graph using SDP-1.
\end{corollary}

This improves upon the previously known guarantee of $O(r/\eps)$ using the Goemans-Linial SDP in \cite{DHV16}, under the same precondition.

\medskip
\noindent{\textbf{Proof Techniques}}:\\
In order to prove our main result, we follow the generic approach of the ARV algorithm~\cite{AroraRV2009} that proceeds in two steps: If there is a dense cluster of the solution vectors, then a specific \Frechet embedding (see Section~\ref{sec:notation} for a definition)  works. If not, then the solutions are `well-spread', and one can always find two $\Omega(n)$-sized sets that are $O(1/\sqrt{\log n})$-apart in $\ell_2^2$ distance, using a separating hyperplane algorithm. This constitutes the core of the proof, and the analysis involves a `chaining argument' which relies on the concentration of measure in high-dimensional spaces. These well-separated sets can then be used to construct a good \Frechet embedding into $\ell_1$.

\smallskip
In our case, we would analogously like to find two large sets that are $\Omega(1/\sqrt{r})$-apart, and to do this, we need to work with the \emph{projections} of the points. Note that the projections need not be in $\ell_2^2$, while the ARV algorithm's analysis requires the use of $\ell_2^2$ triangle inequalities at various  points. 

\smallskip
Thus, in order to prove Theorem~\ref{thm:main-thm}, we follow and adapt the techniques in Naor, Rabani and Sinclair~\cite{NRS05} (henceforth called the NRS analysis). Their work generalized the ARV algorithm's analysis to apply to the more general case of metrics \emph{quasisymmetrically embeddable} into $\ell_2$, which includes $\ell_2^2$ as a special case. We do not need the complete machinery developed by them, though, and extend only a part of their analysis to our setting. In particular, the chaining argument in \cite{NRS05} works in Euclidean, rather than $\ell_2^2$ space, making it useful in our case. 

\smallskip
Our result, thus, also demonstrates the utility of isolating the chaining argument from the use of $\ell_2^2$ triangle inequalities in the ARV algorithm's analysis.

\subsection{Other related Work}
We recall that the best known upper bound for the worst-case distortion of
embedding $\ell_2^2 \hookrightarrow \ell_1$ is $O(\sqrt{\log n} \cdot \log \log n)$ by \cite{AroraLN2008}, building on the techniques in~\cite{AroraRV2009,Lee2005}. The best known lower bound is $\Omega(\sqrt{\log n})$ for worst-case distortion~\cite{NaorY2017}, and
$\exp(\Omega(\sqrt{\log \log n}))$ for average distortion~\cite{KaneM2013}. On \emph{low threshold-rank} graphs (where $\lambda_r \geq \Omega(1) \Phi_{SDP}$), an approximation guarantee of $O(1)$ for Sparsest Cut was obtained using $O(r)$ levels of the Lasserre hierarchy for SDPs~\cite{GuruswamiS2013}. In contrast, the works \cite{DeshpandeV2014,DHV16} obtained a weaker $O(r)$ approximation, but using just the basic SDP relaxation. Oveis Gharan and Trevisan \cite{GharanT2013} also give a rounding algorithm for the basic SDP relaxation on low-threshold rank graphs, but require a stricter pre-condition on the eigenvalues ($\lambda_r \gg \log^{2.5} r \cdot \Phi(G)$), and leverage it to give a stronger $O(\sqrt{\log r})$-approximation guarantee. Their improvement comes from a new structure theorem on the SDP solutions of low threshold-rank graphs being clustered, and using the techniques in ARV for analysis.

Kwok~et al.~\cite{KwokLLGT2013} showed that a better analysis of Cheeger's inequality gives a $O(r \cdot \sqrt{1/\lambda_{r}})$ approximation to the sparsest cut on regular graphs. In particular, when $\lambda_r(G) \geq \epsilon$, this gives a $O(r/\sqrt \epsilon)$ approximation. Note that our result gives a better approximation in this setting (see Section~\ref{sec:sparsestcut}).


\section{Notation}\label{sec:notation}


We use $[n]=\{1,\ldots,n\}$. For a matrix $M \in \mathbb{R}^{d\times d}$, we say $M\succeq 0$ or $M$ is positive-semidefinite (psd) if $y^TXy \geq 0$ for all $y\in \mathbb{R}^d$. The unit Euclidean Ball in $\R^d$ is denoted by $B_2^d$. 

\medskip 
\noindent \textbf{Graphs and Laplacians:} All graphs will be defined on a vertex set $V = [n]$ of size $n$. The vertices will usually be referred to by indices $i,j,k,l \in[n]$. Given a graph with a symmetric weight function on pairs $W:V \times V \mapsto \mathbb{R}^+$, with $W(i,i)=0 \, \forall i$,  let $ D(i) \coloneq \sum_j W(i,j)$ be the degree of vertex $i \in V$. The (normalized) graph Laplacian matrix is defined as:
\begin{align*}
 L_W(i,j) :=
 \begin{cases}
 -\frac{W(i,j)}{\sqrt{D(i)D(j)}} &\quad \text{if $i\neq j$} \\
 1 &\quad \text{if $i=j$}. \\
 \end{cases}
\end{align*}

Note that $L_W \succeq 0$. We will denote the eigenvalues of (the Laplacian of) the graph $G$ by $0 = \lambda_1(G) \leq \lambda_2(G) \ldots \leq \lambda_n(G)$, in \emph{increasing} order. If the graph is $c$-\emph{regular}, we have $D(i) = c$ for every $i \in V$.  Note that $c$ might be a fraction.

\medskip
For nodes $i,j$ in $G$, $d_G(i,j)$ is the shortest path between vertices $i,j$ in $G$. For $S \subseteq [n]$, $G[S]$ is the subgraph induced by $G$ on $S$. The \emph{vertex expansion} of $G$, denoted by $h(G)$ is defined as the largest constant $h$ such that for every set $S \subseteq V$ with $1 \geq |S| \geq |V|/2$, $\abs{N_G(S)} \geq h |S|$ where $N_G(S) = \inbrace{j \in V \suchthat d_G(j,S)=1}$.

\medskip
\noindent \textbf{Embeddings and cuts:} For our purposes, a (semi-)metric space $(X, d)$  consists of a finite set of points $X=\{x_1 , x_2, \ldots, x_n\}$ and a distance function $d: X \times X \mapsto \R_{\geq0}$ satisfying the following three conditions:
\begin{enumerate}
\item $d(x,x)=0$, $\forall x\in X$.
\item $d(x,y)=d(y,x)$.
\item (Triangle inequality) $d(x,y)+d(y,z) \geq d(x,z)$.
\end{enumerate}
An \emph{embedding} from a metric space $(X, d)$ to a metric space $(Y, d')$ is a mapping $f: X \rightarrow Y$. The embedding is called a \emph{contraction}, if
$$d'(f(x_i), f(x_j)) \leq d(x_i, x_j), \qquad\forall x_i, x_j \in X.$$
For convenience, we will only deal with contractive mappings in this paper (this is without loss of generality).
A contractive mapping is said to have (worst-case) distortion $\Delta$, if:  $\sup_{i,j} \frac{d(x_i,x_j)}{d'(f(x_i),f(x_j))} \leq \Delta$.  It is said to have \emph{average} distortion $\beta$, if $\frac{\sum_{i<j}d(x_i,x_j)}{\sum_{i<j}d'(f(x_i),f(x_j))} \leq \beta.$

Note that a mapping with worst-case distortion $\Delta$ also has average distortion $\Delta$, but not necessarily vice-versa.
\smallskip
\emph{\Frechet} embeddings of $(X, d)$ are a class of embeddings of $X \rightarrow \R^k$ into defined on the basis of distances to point sets: a co-ordinate of the embedding will be given by a map of the form $d(x_i, S) \coloneq \min_{j \in S} d(x_i, x_j)$ for some $S\subseteq X$. Note that \Frechet embeddings are always contractive in every co-ordinate.

\medskip
When $X \subseteq \R^k$ is a $\ell_2^2$ space, we will use $d(i,j) \coloneq \norm{x_i-x_j}_2^2$, and $d(S, T) = \min_{i\in S, j\in T} d(i,j)$ for $S, T \subseteq [n]$. For $c\in \R$, $B(i,c) \coloneq \inbrace{j \suchthat d(i,j) \leq c }$. We refer to the quantity $\frac{1}{n^2}\sum_{i,j} \norm{x_i - x_j}_2^2$ as the \emph{spread} of these points.

\section{Proof of Main Theorem}\label{sec:proof-of-main-theorem}

\subsection{Proof Outline}\label{sec:proof-outline}

We prove Theorem~\ref{thm:main-thm} in two steps. First, we scale the points to lie within a $\ell_2$ ball of radius $1$; note that this would shrink the pairwise distances. Suppose that the points have \emph{constant spread} after this scaling; i.e. they satisfy
\begin{equation}\label{eq:const-spread}
\frac{1}{n^2}\sum_{i,j \in V} \norm{x_i - x_j}_2^2 \geq \delta , \qquad \text{where }\delta=\Omega(1). 
\end{equation}

Since scaling does not affect the subspace rank, we continue to have $\ssr_\eta(X) = r$. In this case, we adapt the chaining argument from ~\cite{NRS05} to work on the \emph{projections} $\inbrace{\Pi x_i}_{i \in V}$ to conclude the existence of two large, $\Delta$-separated sets for $\Delta = \Omega(1/\sqrt{r})$.

In the general case, we show that by appropriately utilizing the subspace criterion, we can either reduce it to the case of constant spread, or produce an $O(1)$ distortion \Frechet embedding by considering distances to an appropriate $\ell_2^2$ ball centered at one of the points.

Let $V \coloneq [n]$. We will require the following definitions, following ~\cite{AroraRV2009}:
\begin{definition}[Largeness]
   A subset $A \subseteq V$ is $\beta$-large, if $|A|\geq \beta n$.
   
  \end{definition}
  
\begin{definition}[$\Delta$-separation]
   Subsets $L\subseteq V$ and $R \subseteq V$ are $\Delta$-separated, if $d(L,R)~\geq~\Delta$
\end{definition}

The following lemma, implicit in \cite{AroraRV2009}, gives a sufficient condition for the existence of a  \Frechet embedding into $\ell_1$ with low average distortion.

\begin{lemma}[Sufficient condition] \label{lem:suff-condition}
If there is a set $S\subseteq [n]$ satisfying
\begin{equation}\label{eq:suff-eq}
|S| \sum_{i\notin S} d(i,S) \geq c. n^2
\end{equation}
Then, there is an embedding of the points into $\ell_1$ with  average distortion $1/c$.
\end{lemma}

\begin{proof}
Consider the embedding $i \mapsto d(i,S)$. Clearly, this is a \Frechet embedding, and hence a contraction. Furthermore, we have:
\begin{align*}
\sum_{i,j\in V} \abs{d(i,S) - d(j,S)} & \geq \sum_{i \notin S, j\in S} \abs{d(i,S)-0}\\
&= |S|\sum_{i \notin S} d(i,S) ~
~ \geq c n^2\end{align*} Thus, the average distortion of the map is at most $1/c$.\end{proof}

Note that the existence of two $\Omega(1)$-large, $\Delta$-separated sets $L, R$ would satisfy the above condition, with $S = L$ and $c = O(1/\Delta)$. The above can also be thought of as an embedding into $\ell_2$, since it is one-dimensional.


\subsection{The constant spread case}\label{sec:const-spread-case}

We will start by stating the following Proposition, which is a simple modification of Proposition 3.11 in \cite{NRS05}. Since the proof closely follows the original, requiring only a simple observation, we do not give it here.

\begin{proposition}[From Proposition 3.11 in \cite{NRS05}]\label{prop:euclidean-ball-lemma}
Let $G=(V,E)$ be graph with vertex expansion $h(G)\geq 1/2$.  Let $f:V \rightarrow B_2^d$ be  a mapping that satisfies:
\begin{equation}\label{eq:ell-2-spread}
\frac{1}{n^2}\sum_{i,j \in V} \norm{f(i) - f(j)}_2 \geq \gamma 
\end{equation}

Then, there exists a pair $i,j\in V$, and constants $c_1(\gamma), c_2(\gamma)$ such that 
\begin{equation} 
\norm{f(i)-f(j)}_2 \geq c_1(\gamma) \quad \text{and } \quad d_G(i,j) \leq c_2(\gamma) \sqrt{d}
\end{equation}
\end{proposition}

\begin{remark}
The modification only requires the observation that for any $i , j$ with  $\norm{f(i)-f(j)}_2 \leq c_1(\gamma)$, and $u \suchthat \norm{u}_2 =1$,  $\langle f(i) - f(j) , u \rangle \leq c_1(\gamma)$. This avoids a union bound over the pairs of points in the last step of the proof, the rest of the steps being identical. Combined with the original statement of Proposition 3.11 in \cite{NRS05}, the term $\sqrt{d}$ in the above can be replaced by $\min \inbrace{\sqrt{\log n}, \sqrt{d}}$.
\end{remark}


\medskip
We now proceed to prove a special case of Theorem~\ref{thm:main-thm} assuming condition \eqref{eq:const-spread}.

\begin{theorem}\label{thm:const-spread}
Let  $X= \inbrace{x_1, \ldots, x_n}$ satisfy $\ell_2^2$-triangle inequalities, with $X\subseteq B_2^d$ and $\ssr_\eta(X) = r$. Furthermore, suppose that $$\frac{1}{n^2}\sum_{ij} \norm{x_i - x_j}_2^2 \geq \delta , \qquad \text{where }\delta=\Omega(1).$$
Then there exist sets $A, B \subseteq X$, with $|A|, |B| \geq (\eta \delta/32)n$ with $d(A,B) \geq \Omega(1/\sqrt{r})$.
\end{theorem}

\begin{proof}
Let $\Pi$ be the $r$-dimensional subspace containing an $\eta$ fraction of the squared lengths of the difference vectors upon projection. Let $V = [n]$, and define $f:V\rightarrow B_2^r$ by 
\[
f(i) \eqdef \Pi x_i
\]
Since the set $X$ has $\eta$-subspace rank $r$, we have, by definition:
\begin{equation}\label{eq:proj-spread}
\frac{1}{n^2} \sum_{i,j\in V} \norm{f(i)-f(j)}_2^2 \geq \eta\delta .
\end{equation}

We will now follow the proof of Theorem 2.4 in \cite{NRS05}, but switch to the projections where appropriate. Consider the graph $G=(V,E)$ with edges  $E=\inbrace{\inbrace{i,j} \suchthat \norm{x_i - x_j}_2^2 \leq \frac{\kappa}{\sqrt r}}$, where $\kappa = \kappa(\eta, \delta)$ is a constant that we will set later.
\smallskip

Suppose, for the sake of contradiction, that every two sets $A, B \subseteq V$ with $|A|, |B| \geq (\eta \delta /32)n$ satisfy $d(A,B) \leq \kappa/\sqrt{r}$, which implies that $d_G(A,B) \leq 1$.  We use the following lemma from ~\cite{NRS05}:

\begin{lemma}[Lemma 2.3 in \cite{NRS05}] \label{lem:vertex-expansion}
Fix $0<\eps \leq \frac{1}{10}$, and let  $G=(V,E)$ be a graph such that for every $X,Y\subseteq V$ satisfying $|X|, |Y| \geq \eps |V|$, $d_G(x,y)\leq 1$. Then there is a $U\subseteq V$ with $\abs{U} \geq (1-\eps)|V|$ with $h(G[U])\geq \frac{1}{2}$.
\end{lemma}

Invoking Lemma~\ref{lem:vertex-expansion} on $G$ yields a subset $X' \subseteq V$, with $|X'| \geq (1-\frac{\eta \delta}{32}) n$  such that $h(G[X'])\geq \frac{1}{2}$. We claim the following:
\begin{equation}\label{eq:x-spread}
\frac{1}{|X'|^2} \sum_{i,j\in X'} \norm{f(i)-f(j)}_2  \geq \frac{(\eta \delta)^{3/2}}{32}. 
\end{equation}

To see this, note that $\abs{X' \times X'} \geq (1-\frac{\eta\delta}{16}) n^2$.  Let $D = \inbrace{(i,j)\in V\times V \suchthat \norm{f(i)-f(j)}_2^2 \geq \eta \delta /4 }$. Since the diameter of the unit ball is $2$,  in order to satisfy \eqref{eq:proj-spread}, we should have $|D| \geq (\eta \delta/8)n^2$.  Thus, $|D \cap (X' \times X')| \geq \frac{\eta\delta}{16} n^2$. This implies that the average $\ell_2$-distance in $X' \times X'$ is at least:
\begin{equation}
\frac{1}{n^2} |D \cap (X' \times X')| \times \sqrt{\frac{\eta \delta}{4}} \geq \frac{(\eta \delta)^{3/2}}{32}.  
\end{equation}
This proves \eqref{eq:x-spread}.

\smallskip
We can now apply Proposition~\ref{prop:euclidean-ball-lemma} to $G[X']$, and the projections $\inbrace{f(i)}_{i\in V}$, with $\gamma = (\eta \delta)^{3/2}/{32}$.  We infer that there exists a path in $G$, of   $k \leq  c_2(\gamma) \sqrt{r} = a(\eta, \delta) \sqrt{r}$ vertices ${i_1, i_2, \ldots i_k} \subseteq X'$ such that $\norm{f(i_1) - f(i_k)}_2 \geq c_1(\gamma) = b(\eta, \delta)$, where $a(\eta, \delta)$ and $b(\eta, \delta)$ are constants depending on $\eta$ and $\delta$.

\smallskip

This implies that:
\begin{equation}
b^2(\eta, \delta) \, \overset{(a)}{\leq} \, \norm{f(i_1) - f(i_k)}_2^2 \, \overset{(b)}\leq \, \norm{x_{i{_1}} - x_{i_{k}}}_2^2 \, \overset{(c)}{\leq}\, \sum_{j=1}^{k-1} \norm{x_{i_j} - x_{i_{j+1}}}_2^2 \, \overset{(d)}{\leq}\, a(\eta, \delta) \sqrt{r}\frac{\kappa}{\sqrt{ r}}.
\end{equation}    

Above, $(b)$ follows from the fact that projections can only decrease distances, $(c)$ from the $\ell_2^2$ property, and $(d)$ from the definition of $G$. This is a contradiction, if we set $\kappa <\frac{b^2(\eta, \delta)}{a(\eta, \delta)}$.
\end{proof}

\begin{remark}
The last chain of inequalities above is the only place where the $\ell_2^2$ triangle inequalities are invoked. Without them, we could still prove a weaker statement with $O(1/r)$ separation between the large sets, since $(c)$ would hold with an additional multiplicative factor of $k$ by convexity.
\end{remark}

\subsection{The general case}\label{sec:general-case}
We now extend our argument to the general case.  Let us fix some notation before going to the proofs. We will take $V \coloneq [n]$,  and $X = \inbrace{x_1, \ldots, x_n}$ to satisfy the $\ell_2^2$ triangle inequalities, with $\ssr_\eta(X) = r$. Let $\Pi$ be the corresponding $r$-dimensional subspace. Let $f(i) \coloneq \Pi x_i$, as before.  Define
\[
d_f(i,j) \coloneq \norm{f(i)- f(j)}_2^2
\]
The terms $d_f(i,S)$,  $d_f(S, T)$ for $S, T \subseteq V$ are defined naturally, and denote $\diam_f(S) \eqdef \max_{i,j \in S} d_f(i,j)$. Note that $d_f(\cdot, \cdot)$ is \emph{not necessarily} a distance, unlike $d(\cdot, \cdot)$. However, since $f$ is a projection map,  it satisfies:
\begin{equation} \label{eq:proj-leq-real}
d(i,S) \geq d_f(i,S) \qquad \forall i\in V, \, \forall S\subseteq V,
\end{equation}

\noindent We will also assume that $X$ is scaled to satisfy:
\begin{equation}\label{eq:spread}
\frac{1}{n^2} \sum_{i,j \in V} \norm{x_i - x_j}_2^2 \,= \, 1
\end{equation}

\smallskip
\noindent We first record a simple observation.

\begin{observation}\label{obs:no-triang}
For any $i,j \in V$, and any $S\subseteq V$,
\[d_f(i,j) \leq 3 \, (d_f(i,S) +  \diam_f(S) + d_f(j,S)).\]
\end{observation}
\begin{proof}
Let $i^*, j^* \in S $ be such that $d_f(i,S) = d_f(i, i^*)$ and $d_f(j,S)= d_f(j, j^*)$. Since $\sqrt{d_f}$ obeys the triangle inequality, we have:
\begin{align*}
\inparen{\sqrt{d_f(i,j)}}^2 &\leq  \inparen{\sqrt{d_f(i, i^*)}+ \sqrt{d_f(i^*, j^*)} + \sqrt{d_f(j, j^*)}}^2 \\
& \leq 3(d_f(i,S) +  \diam_f(S) + d_f(j,S)) 
\end{align*}
The last inequality follows from the convexity of the function $g(x) = x^2$, and the definition of $\diam_f$.
\end{proof}

\noindent We now consider various cases, and show that a low average-distortion embedding exists in each case.

\begin{lemma}[Dense Ball]\label{lem:dense-ball-lemma}
If $\exists i\in V$, with $|B(i, 1/12)| \geq n/12$, then we can find an $O(1)$-average distortion embedding of $X$ into $\ell_1$.
\end{lemma}
\begin{proof}
The proof follows the proof of a similar lemma in \cite{AroraRV2009}. Let $i_0 \in V$ be such that $|B(i_0, 1/12)| \geq n/12$, and let $S = B(i_0, 1/12)$. Consider the embedding $ i \mapsto d(i,S)$. This is a contraction. Since $\sum_{ij} \norm{x_i - x_j}_2^2 = n^2$, we have : 
\begin{align*}
n^2 &= \sum_{i,j \in V} d(i,j) \\
& \leq \sum_{i,j \in V} \inparen{d(i,S)+ d(j,S)}  \qquad \ldots \text{ Using $\ell_2^2$ triangle inequality }\\
& = 2n \inparen{ \sum_{i \notin S} d(i,S)} 
\end{align*}

This gives us that $\sum_{i \notin S} d(i,S) \geq n/12$. Since $|S| = \Omega(n)$, Lemma~\ref{lem:suff-condition} applies, and proves that the above embedding has $O(1)$ average-distortion.  \footnote{Strictly speaking, one could do without the $\ell_2^2$ triangle inequality here by adjusting the constants appropriately, as we did in Observation~\ref{obs:no-triang}.}
\end{proof}

\begin{lemma}[Isolating a bounded ball] \label{lem:bounded-ball}
If there is no $i\in V$ such that $|B(i,1/12)| \geq n/12$, then there is a $j \in V$ such that $S=B(j, 12/9)$ satisfies  $|S| \geq \frac{3}{12}n$, and
 \[
 \sum_{i,j \in S} d(i,j) \geq \inparen{\frac{2}{12}} \inparen{\frac{1}{12}} \frac{n^2}{12}
 \]
\end{lemma}
\begin{proof}
Suppose we had $|B(j, 12/9)| < (3n/12)$ for every $j\in V$. Then, for any $j\in V$, we would have $|\overline{B(j, 12/9)}| > 9n/12$, which  gives us that $\sum_i d(j, i) > n$. Summing over $j \in V$ contradicts \eqref{eq:spread}.

\smallskip
\noindent Now, let $j_0 \coloneq \arg \max_{j\in V} |B(j,12/9)|$, and $S \coloneq B(j_0 ,12/9)$.
 Define the set $A= B(j_0,12/9) \setminus B(j_0, 1/12)$. From our assumption and the preceeding argument, $|A| \geq 2n/12$. Since $|B(i,1/12)| \leq n/12$ for every $i \in A$, we have that $\abs{\overline{B(i,1/12)} \cap A} \geq n/12$. This gives us:
 
$$ \sum_{i\in A,  j\in A} d(i,j) \geq \frac{2n}{12} \times \frac{1}{12} \times \frac{n}{12} $$
\end{proof}

\medskip

In next two lemmas, assume that the precondition of Lemma~\ref{lem:bounded-ball} holds, i.e., there is no $i\in V$ with   $|B(i,1/12)| \geq n/12$.

\begin{lemma} \label{prop:ball-projection-good}
Let $j_0 = \arg \max_{j\in V} |B(j, 12/9)|$, and $S \eqdef B(j_0, 12/9)$. If $S$ satisfies:
\[
\sum_{i,j \in S} d_f(i,j) \geq \frac{\eta}{600} |S|^2,
\]
then there is an embedding of $X$ into $\ell_1$ with $O(\sqrt{r})$ average distortion.
\end{lemma}

\begin{proof}
Consider the map $g:V \rightarrow \R^d$ given by $g(i) \eqdef \sqrt{9/12}\cdot x_i$. This ensures that $g(i) \in B_2^d$ for every $i\in S$, and the mapping continues to obey the $\ell_2^2$ triangle inequalities. Furthermore, from Lemma~\ref{lem:bounded-ball}, the points in $S$ satisfy:
\begin{equation}\label{eq:s-spread-proj}
\frac{1}{|S|^2} \sum_{i,j \in S } \norm{g(i)-g(j)}_2^2 \geq \frac{9}{12} \times \frac{2}{12^3} = \Omega(1)
\end{equation}

\noindent From the assumption on $S$, we infer that:
\[
\frac{1}{|S|^2} \sum_{i,j \in S } \norm{\Pi g(i) - \Pi g(j)}_2^2 \,\geq  \, \frac{9}{12} \times \frac{\eta}{600}
\]

\noindent We can now invoke Theorem~\ref{thm:const-spread} on just the points in $S$ to conclude that there exist sets $A, B \subseteq S$, such that $|A|, |B| \geq \Omega_\eta(n)$ with $d(A,B) \geq \Omega_\eta(1/\sqrt{r})$ (the scaling by a constant factor just shrinks some distances). As before, it is easy to see that $A$ satisfies the conditions of Lemma~\ref{lem:suff-condition} with $c= \Omega(1/\sqrt{r})$ and hence the mapping $h(i)\eqdef d(i,A)$ has average distortion $O(\sqrt{r})$. Note that by the ARV algorithm~\cite{AroraRV2009}, the sets can be found with good probability by a random separating hyperplane through $j_0$.
\end{proof}

\medskip 
\begin{lemma}\label{prop:ball-projection-bad}
Let $j_0 = \arg \max_{j\in V} |B(j, 12/9)|$, and $S \eqdef B(j_0, 12/9)$. If $S$ satisfies:
\[
\sum_{ij \in S} d_f(i,j) \leq \frac{\eta}{600} |S|^2 ,
\]
then we can find an embedding of $X$ into $\ell_1$ with $O(1)$ average distortion.
\end{lemma}

\begin{proof}
The proof will be similar to the proof of Lemma~\ref{lem:dense-ball-lemma}, except for the fact that we will work with projections instead of the original vectors.
\smallskip
First, observe that there exists an $i_0 \in S$ such that $|B_f(i_0, \eta/24) \cap S| \geq 24|S|/25$. If not, then for every $i\in S$, we will have $\sum_{j \in S} d_f(i,j) > \frac{1}{25}|S| \times \eta/24 = \eta|S|/600$. Summing over $j \in S$ results in a contradiction to the precondition on $S$.  

\smallskip
\noindent Let $T \eqdef  B_f(i_0, \eta/24)$; from the preceding argument, we have $|T| = \Omega(n)$.
\smallskip

\begin{claim} \label{cl:T-dense-ball}
$
\sum_{j\notin T } d_f(j, T) \geq \eta n/12$
\end{claim}
\begin{proof}
We know that $\sum_{i,j \in V} \norm{f(i) - f(j)}_2^2 = \sum_{i,j \in V} d_f(i,j) \geq \eta n^2$. Using Observation~\ref{obs:no-triang}, we can infer:
\begin{align*}
\eta n^2 &\leq  \sum_{i,j \in V} d_f(i,j) \\
& \leq 3~\sum_{i,j \in V} \inparen{d_f(i,T) + \diam_f(T) + d_f(j,T)} \qquad \ldots \text{Using Observation~\ref{obs:no-triang}}\\
&= 3~\inparen{2n~\sum_{i\in V} d_f(i,T) \,+\, \frac{4\eta}{24}n^2} \qquad \ldots \text{ Since } \diam_f(T) \leq \frac{4\eta}{24} \\
\end{align*}
This yields that $\sum_i d_f(i,T) \geq \frac{\eta}{12} n$, proving the claim. 
\end{proof}

\smallskip

Since $|T| = \Omega(n)$, and $d(i,T) \geq d_f(i,T)$, $T$ satisfies the conditions of Lemma~\ref{lem:suff-condition}. This gives us an $O(1)$ average-distortion embedding of the points into $\ell_1$.
\end{proof}

We can now infer the proof of Theorem~\ref{thm:main-thm} by using the results above.

\begin{proof}[Proof of Theorem~\ref{thm:main-thm}]
The conditions covered in Lemmas~\ref{lem:dense-ball-lemma},~\ref{lem:bounded-ball},~\ref{prop:ball-projection-good} and \ref{prop:ball-projection-bad} on the set of points $\{x_i\}_{i \in V}$ are exhaustive, and in each case yield an embedding with $O(\sqrt{r})$ average distortion. It is clear that each of these conditions can be easily checked, and the corresponding embeddings can be constructed efficiently. 
\end{proof}

\begin{remark}
The Hamming Cube on $N$ points, residing in $\log N$ dimensions, and having $\eta$-subspace rank $\Omega_\eta(\log N)$ by symmetry, has two $\Omega(N)$-sized sets that are $\Omega(1/\sqrt{ \log N})$ apart, and shows that the above analysis is tight up to constants.
\end{remark}

\subsection{Application to Sparsest Cut}~\label{sec:sparsestcut}
\noindent The proof of Corollary~\ref{cor:main-cor} now follows easily, using the main result.

\begin{proof}[Proof of {Corollary~\ref{cor:main-cor}}]
\smallskip
Suppose $\lambda_r/n \geq \Phi_{SDP}/(1-\epsilon)$. We invoke the following result of Guruswami and Sinop~\cite{GuruswamiS2013} (stated here for the special case of \emph{Uniform} Sparsest Cut): 

\begin{proposition}[Von-Neumann inequality~{\cite[Theorem~3.3]{GuruswamiS2013}}]\label{prop:laplacian}
Let $\sigma_{1} \geq \sigma_{2} \geq \dotsc \geq \sigma_{n} \geq 0$ be the singular values of the matrix $M$ with columns $\{(x_{i} - x_{j})\}_{i<j}$. Then
$$
\frac{\sum_{t \geq r} \sigma_{j}^2}{\sum_{t=1}^{n} \sigma_{j}^2} \leq \frac{\Phi_{SDP}}{\lambda_{r}(G)/n}.
$$
\end{proposition}

For every $l \leq n$, we know that $\sum_{i=1}^{l} \sigma_i^2 = \sum_{i<j} \norm{\Pi_{l} (x_i - x_j)}_2^2$, where $\Pi_{l}$ is the subspace defined by the the top $l$ left singular vectors of $M$. This immediately gives us that $\ssr_{\epsilon}(X) = r-1$. Applying the main theorem gives us an $O(\sqrt{r})$ average distortion embedding into $\ell_1$, and hence an $O_\epsilon(\sqrt{r})$ approximation to $\Phi(G)$ in this setting.
\end{proof}

\begin{remark}
Under the same precondition, Guruswami and Sinop~\cite{GuruswamiS2013} give an $O(1/\epsilon)$ approximation, but by solving a SDP of size $n^{O(r)}$, using a partial solver that runs in time $2^{O(r)}\mathrm{poly}(n)$~\cite{GuruswamiS2012b}. They need to know $r$ first, and set up the SDP and solver appropriately. The works \cite{DeshpandeV2014, DHV16} give a $O(r/\epsilon^2)$ and $O(r/\epsilon)$ approximation respectively, using  just the Goemans-Linial SDP; the rounding algorithms do not depend on $r$. Our algorithm too is independent of $r$, and we get a better guarantee of $O(\sqrt{r}/\epsilon)$ in this setting. 
\end{remark}

Though the precondition of the corollary may seem involved, it can easily be related back to a simpler one, as the following corollary shows (proof in Appendix~\ref{sec:app-proofs}).

\begin{corollary}\label{cor:sp-cut-cor}
If $G$ is regular with $\lambda_r(G) \geq \epsilon$, then we can find a $O(\sqrt{r} \,+\, 1/\sqrt{\epsilon})$ approximation to the sparsest cut in $G$ in $\mathrm{poly}(n)$ time. 
\end{corollary}

\begin{remark}
It is clear that we get a $O(\sqrt{r})$ approximation for all graphs whose $\ell_2^2$ representation always has subspace rank $r$. Graphs of low threshold-rank are one class of graphs that have this property.
\end{remark}

\subsubsection*{Acknowledgements}
The second named author would like to thank Amit Deshpande and Prahladh Harsha for prior useful discussions.


\bibliographystyle{plainurl}
\bibliography{lde-bib}

\appendix
\section{Appendix}

\subsection{Ruling out a  worst-case distortion bound of $O(\sqrt{\ssr_\eta(X)})$.}\label{sec:no-worst-case}

We give a simple example of why one cannot hope to prove a worst-case distortion bound like Goemans' result, using the notion of subspace rank. Suppose that a certain point set $\emph{X}$ satisfies the $\ell_2^2$ inequalities, and has \emph{worst-case} distortion $\Omega(D)$ for embedding into $\ell_1$. It is known that there exists such an $X$ with $D = \Omega(\sqrt{\log n})$~\cite{NaorY2017}. Without loss of generality, let $X$ be scaled to satisfy $\sum_{i,j} \norm{x_i - x_j}_2^2 = n^2$, and $\norm{x_1 - x_2}_2^2 = \max_{i,j} \norm{x_i - x_j}_2^2 $. Consider the set $Y$ which has $X$, along with $C-1$ additional copies of $x_1$ and $x_2$\footnote{Technically, we are dealing with semi-metrics, and hence distinct points may overlap.}. Clearly, $Y$ satisfies the $\ell_2^2$ triangle inequalities. Further, $Y$ has $\eta$-subspace rank of $1$ for a large enough $C$: the sum of all squared distances is at most $C + (C^2-C) \norm{x_1 - x_2}_2^2$, and the sum of squared distances along the direction $x_1- x_2$ is  at least $C^2 \norm{x_1 - x_2}_2^2$.  However, embedding $Y$ with \emph{worst}-case distortion $O(1)$ into $\ell_1$ would contradict the lower bound on embedding $X$ into $\ell_1$.

\subsection{Proof of Corollary~\ref{cor:sp-cut-cor}} \label{sec:app-proofs}

\begin{proof}[Proof (Of Corollary~\ref{cor:sp-cut-cor})]
The proof follows by using a combination of two algorithms, depending on how $\lambda_r$ compares to $\Phi_{SDP}(G)$.
\smallskip
Suppose that $G$ is $1$-regular by scaling the edge weights, without loss of generality, and let $X=\inbrace{x_1, \ldots, x_n}$ be the optimal SDP solution. If $\Phi_{SDP} \geq \epsilon/100n$, then there is one co-ordinate of the SDP solution with objective value at least $\epsilon /100n$.
In this case, running the Cheeger rounding algorithm ~\cite[Lemma 2.1]{AlonM1985} (see also ~\cite[Section 2.4]{Trevisan2011} for an exposition) on this co-ordinate would output a cut of sparsity $O(\sqrt{\epsilon}/n) \leq O\inparen{\Phi_{SDP}(G)/\sqrt{\epsilon}}$.

\smallskip
If $\Phi_{SDP} \leq \epsilon /100 n$ then we have $\lambda_r/n \geq 100 \Phi_{SDP}$. Applying Corollary~\ref{cor:main-cor} with $\epsilon = 99/100$ gives us an $O(\sqrt{r})$ average-distortion embedding into $\ell_1$, and hence an $O(\sqrt{r})$ approximation to $\Phi(G)$ in this setting. Thus, the best of the two cuts will be a $ O(\sqrt{r}+ 1/\sqrt{\epsilon})$ approximation to $\Phi(G)$.
\end{proof}

\end{document}